\documentclass[12pt]{article}

\usepackage[english]{babel}

\usepackage[letterpaper,top=2cm,bottom=2cm,left=3cm,right=3cm,marginparwidth=1.75cm]{geometry}

\usepackage{amsmath}
\usepackage{graphicx,natbib}
\usepackage[colorlinks=true, allcolors=blue]{hyperref}
\usepackage{multirow}
\usepackage{graphics,booktabs}
\usepackage{amsthm,amssymb,authblk}
\usepackage{doi}

\theoremstyle{definition}
\newtheorem{theorem}{Theorem}[section]
\newtheorem{remark}{Remark}[section]

\usepackage{algorithm,algpseudocode}

\title{Subdata selection for big data regression: an improved approach}
\author{Vasilis Chasiotis}
\author{Dimitris Karlis}

\affil{\small Department of Statistics, Athens University of Economics and Business,  Greece}

\date{}

\begin{document}

\maketitle
    
\begin{abstract}
In the big data era researchers face a series of problems. Even standard approaches/methodologies, like linear regression, can be difficult or problematic with huge volumes of data. Traditional approaches for regression in big datasets may suffer due to the large sample size, since they involve inverting huge data matrices or even because the data cannot fit to the memory. Proposed approaches are based on selecting  representative subdata to run the regression. Existing approaches select the subdata using information criteria and/or properties from orthogonal arrays. 
In the present paper we improve existing algorithms providing a new algorithm that is based on D-optimality approach. We provide simulation evidence for its performance.
 Evidence about the parameters of the proposed algorithm is also provided in order to clarify the trade-offs between execution time and information gain. Real data applications are also provided.
\end{abstract}


\section{Introduction}
\label{introduction}
Recent research in various disciplines is characterized by the unprecedented demand of big data analysis. Typically, the scale of the datasets increases, so does the demand of computational resources for the statistical analysis and modeling process. Although the availability of computational power increases rapidly, it still falls far behind under the explosive increase in data volume. 

This creates new challenges to data storage and analysis. A standard approach is based on data reduction, or subsampling, where one selects a portion of the data to extract useful information. This is a crucial step in big data analysis. For massive data, subsampling techniques are popular to mitigate computational burden by reducing the data size considerably and bringing it back to a doable size.
Consider the problem of regression where the sample size $n$ is quite large and one needs to fit a model with standard least squares approach. In many circumstances, this can create a lot of computational issues, since the standard least square approach involves big matrices that perhaps do not fit in the memory. Working with less data is an option as far as the reduced dataset can keep as much information as possible. In most problems, while picking the necessary data with pure randomness is an option, improved approaches can be used to select subdata in an optimal way. 

As a first attempt, \cite{drineas2011faster} proposed the idea of selecting part of the data with a random selection. In particular, they proposed to make a randomized Hadamard transform on data and then use uniform subsampling to take random subdata to approximate ordinary least squares estimators in linear regression models. Such an approach, based on the idea and theory of random matrices, has found a lot of applications. On the other hand, they suffer from the inherent randomness of the procedure. In recent years, an alternative approach attempts to use deterministic rather than random selection of data points based on certain criteria. Such approaches share significant relationship with optimal-design problems, traditionally known in statistics for many years. Methods of optimal design of experiments might also be applicable to the setting of big data by providing the methodology for selecting data points to create the optimal subdata. 

The paper of \cite{wang2019information} is central to such approaches bringing ideas from optimal designs to the selection of data points proposing the information-based optimal subdata selection (IBOSS) approach. Recent extensions are given by \cite{wang2021oss} with the orthogonal subsampling (OSS) approach and some other ideas with other loss functions in \cite{ren&zhao}. Details about the aforementioned approaches will follow in Section \ref{back}.

A classical optimization problem is to find an ellipsoid of the smallest volume that encloses a given dataset centered around the origin, i.e., the minimum-volume enclosing ellipsoid (MVEE). The MVEE finds applications across a variety of fields, including statistics. One notable application of MVEE in statistics is in the optimal design of experiments. The approximate D-optimal design problems and the MVEE are equivalent; to be precise, they are dual problems. For more information, readers can refer to \cite{silvey1973}, \cite{titterington1975}, \cite{ahipasaoglu2015} and \cite{harman2020}. More recently, \cite{rosa2022} proposed an algorithm for calculating the MVEE, particularly useful for large datasets stored in a separate database. The algorithm developed by \cite{rosa2022} addresses scenarios where existing algorithms may encounter memory limitations or be extremely time-consuming. In \cite{ahipasaoglu2015}, it is noted that ellipsoids can be used to approximate a dataset for three reasons. One of these reasons is that ellipsoids are both simple and flexible for estimating various properties of a dataset, including the volume of its convex hull. Therefore, a plausible idea for selecting subdata, is to select data points in some sense with large convex hull as close as possible to the one generated by the full data. In such a case, the selected data points can have a large volume and hence the determinant of the information matrix will be large. Recall that the inverse of the information matrix appears in the variance of the estimated coefficient vector

To motivate the problem, consider the data in Figure \ref{mot}. We have used two covariates and $200$ data points. Suppose that we want to select $8$ observations. In the first subplot, the shaded area represents the convex hull of the full data, which is generated by the $9$ data points highlighted in red. In the remaining subplots, the shaded area represents the convex hull generated by the selected subdata, also highlighted in red.  The IBOSS approach (second panel) tries to select points at the extreme of the two covariates, whereas the OSS approach (third panel), using a different loss function, attempts to select data points at the corners of the two-dimensional space. A detailed description of the aforementioned approaches is given in Section \ref{back}. 
Our aim and contribution of the present paper, is to improve and create an approach (last panel) that balances the two ideas and with a few steps can improve a lot by selecting data points by the D-optimality criterion, namely to increase the determinant of the information matrix. In this selected motivating example, our approach is more successful than the other approaches in approximating the convex hull of the full data. Note that we present a main algorithm that starts from an already existing approach and with a few additional considerations, it improves with respect to the generalized variance of the selected subdata. We also propose some extensions which at the cost of additional execution time can further improve the D-optimality criterion under the selected data points quite a lot to balance the additional time needed. 

\begin{figure}[!thb]
\begin{center}
  \includegraphics[width=1\textwidth]{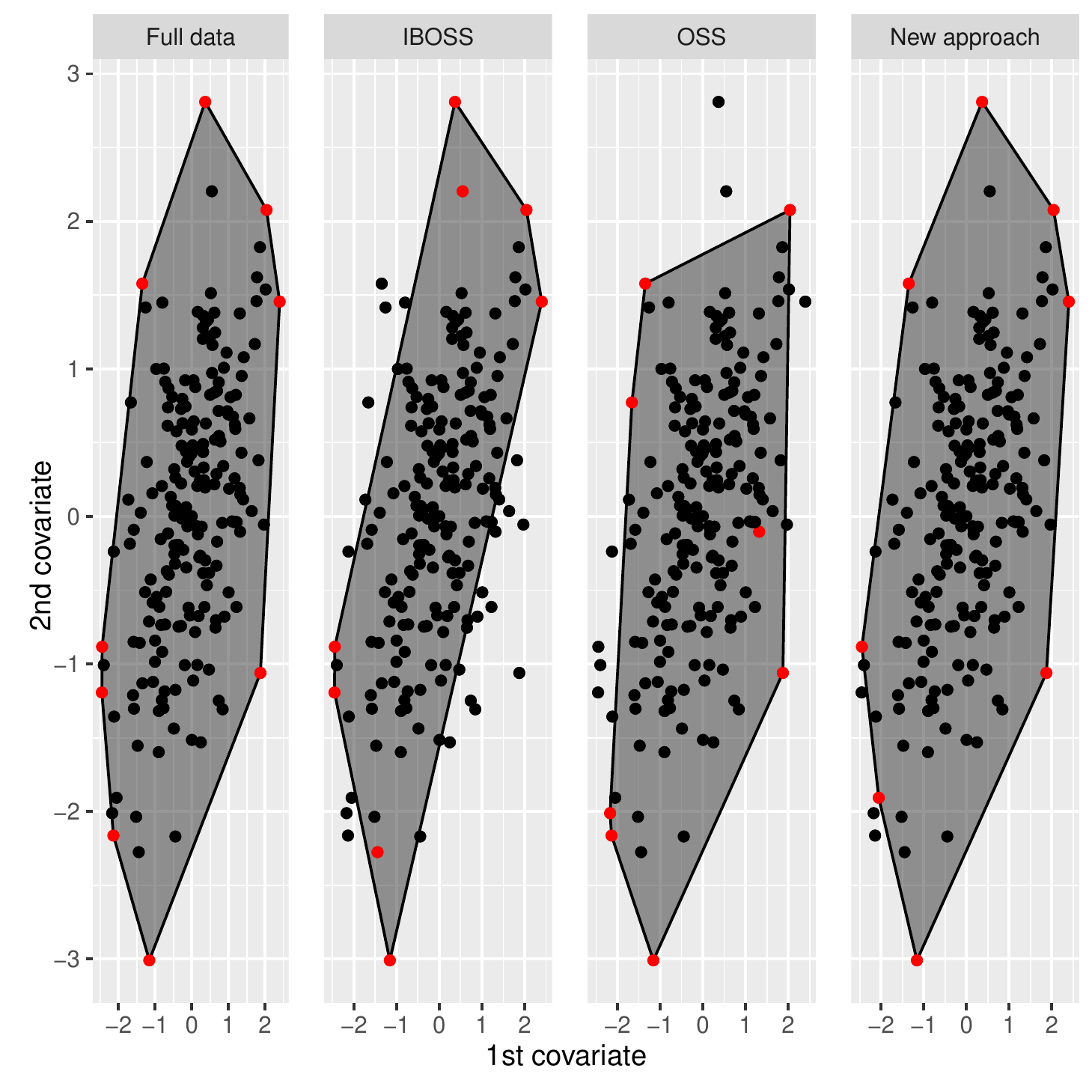}   
\end{center}
 \caption{An example for the different approaches. Data with two covariates and full data size of $200$ data points were generated. The different approaches were used to select $8$ data points that are highlighted in red. The convex hull of the full data can be seen in the first panel. In the remaining panels, one can see the convex hull generated by the selected subdata by the IBOSS approach, the OSS approach and our new approach.}
 \label{mot}
\end{figure}

The rest of the paper proceeds as follows.
Section \ref{theor} provides some theoretical arguments that will be the basis of the newly developed approach. Section \ref{back} describes approaches already existing in the current literature. The algorithm and its variants are described in detail in Section \ref{section_algorithm}. Simulation evidence to support the new approach is provided in Section \ref{section_simulation}. It includes a comparison with the existing approaches to show how and when the proposed approach improves with respect the existing ones. Two real datasets are used for illustration in Section \ref{section_data}, while concluding remarks can be found in Section \ref{concluding_remarks}.

\section{Theoretical considerations}\label{theor}

We assume that the full data are denoted by ($\textbf{x}_i, y_i$), $i=1,2,\ldots,n$. Given the following linear regression model:
\begin{equation}\label{model1}
y_i=\beta_0+\textbf{x}_i^\text{T}\boldsymbol{\beta}_1+\epsilon_i, \quad i=1,2,\ldots,n,
\end{equation}
where $y_i$ is a response, $\beta_0$ is the intercept parameter, $\textbf{x}_i=(x_{i1},x_{i2},\ldots,x_{ip})^{\text{T}}$ is a covariate vector, $\boldsymbol{\beta}_1=(\beta_1,\beta_2,\ldots,\beta_p)^{\text{T}}$ is a $p$-dimensional vector of unknown slope parameters, and $\epsilon_i$'s are the error terms that are uncorrelated satisfying $E(\epsilon_i)=0$ and $V(\epsilon_i)=\sigma^2$. 

We take into consideration the full data under model \eqref{model1}, and so the least-square estimator of $\boldsymbol{\beta}=(\beta_0,\boldsymbol{\beta}_1^{\text{T}})^{\text{T}}$, which is its best linear unbiased estimator as well, is
\begin{equation*}
\hat{\boldsymbol{\beta}}_{\text{Full}}=\left(\sum_{i=1}^{n}\textbf{z}_i\textbf{z}_i^{\text{T}}\right)^{-1}\sum_{i=1}^{n}\textbf{z}_iy_i,
\end{equation*}
where $\textbf{z}_i=(1,\textbf{x}_i^\text{T})^\text{T}$.

The covariance matrix of $\hat{\boldsymbol{\beta}}_{\text{Full}}$ is equal to the inverse of
\begin{equation*}
\textbf{Q}_{\text{Full}}=\dfrac{1}{\sigma^2}\sum_{i=1}^{n}\textbf{z}_i\textbf{z}_i^{\text{T}},
\end{equation*}
where $\textbf{Q}_{\text{Full}}$ is the observed Fisher information matrix of $\boldsymbol{\beta}$ for the full data if $\epsilon_i$ are normally distributed. Even though we do not require the normality assumption, $\textbf{Q}_{\text{Full}}$ will be still called the information matrix.

However, it is not always feasible to fully analyze the whole data, since the sample size $n$ of the full data can be too large. Thus, an approach is to gain useful information from the full data given that computational resources are limited. An effective investigation can be focused on selecting a subset of the full data.

Let $\delta_i$ be a variable that indicates whether ($\textbf{x}_i, y_i$) is included in the subdata. Therefore, $\delta_i=1$ if ($\textbf{x}_i, y_i$) is included in the subdata and $\delta_i=0$ otherwise. In case of selecting subdata of size $k$, the least-square estimator of $\boldsymbol{\beta}$ remains the best linear unbiased estimator based on the subdata, that is, 
\begin{equation*}
\hat{\boldsymbol{\beta}}_{\text{Sub}}=\left(\sum_{i=1}^{n}\delta_i\textbf{z}_i\textbf{z}_i^{\text{T}}\right)^{-1}\sum_{i=1}^{n}\delta_i\textbf{z}_iy_i,
\end{equation*}
where $\sum_{i=1}^{n}\delta_i=k$.

The information matrix with subdata of size $k$ can be written as
\begin{equation}\label{inf_sub}
\textbf{Q}_{\text{Sub}}=\dfrac{1}{\sigma^2}\sum_{i=1}^{n}\delta_i\textbf{z}_i\textbf{z}_i^{\text{T}}.
\end{equation}
The selected subdata should be optimal is some way. According to the theory of the optimal design of experiments, this is assessed with respect to a criterion, which is related to the covariance matrix of the estimated parameters. Such a popular criterion is the D-optimality criterion which seeks to maximize the determinant of the information matrix. Thus, as \cite{wang2019information} proposed, we are interested in maximizing the determinant of $\textbf{Q}_{\text{Sub}}$ subject to $\sum_{i=1}^{n}\delta_i=k$.

However, the information matrix $\textbf{Q}_{\text{Sub}}$ in \eqref{inf_sub} can be written in another way, and so the problem of maximizing its determinant can be finally redefined. The information matrix $\textbf{Q}_{\text{Sub}}$ in \eqref{inf_sub} can be written as
\begin{equation*}
\textbf{Q}_{\text{Sub}}=\dfrac{k}{\sigma^2}\begin{bmatrix}
	1 & \dfrac{\sum_{i=1}^{n}\delta_ix_{i1}}{k} & \cdots & \dfrac{\sum_{i=1}^{n}\delta_ix_{ip}}{k}\\ 
	\dfrac{\sum_{i=1}^{n}\delta_ix_{i1}}{k} & \dfrac{\sum_{i=1}^{n}\delta_ix_{i1}^2}{k} & \cdots & \dfrac{\sum_{i=1}^{n}\delta_ix_{i1}x_{ip}}{k} \\
	\vdots & \vdots & \ddots & \vdots \\
	\dfrac{\sum_{i=1}^{n}\delta_ix_{ip}}{k} & \dfrac{\sum_{i=1}^{n}\delta_ix_{i1}x_{ip}}{k} & \cdots & \dfrac{\sum_{i=1}^{n}\delta_ix_{ip}^2}{k}
\end{bmatrix}.
\end{equation*}

In order to discriminate between the covariate vectors $\textbf{x}_i$, $i=1,2,\ldots,n$ and the covariates, let $\textbf{x}_j^*$, $j=1,2,\ldots,p$ be the $j$th covariate under the selected subdata. After some calculations, we get that:
\begin{equation*}
\textbf{Q}_{\text{Sub}}=\dfrac{k}{\sigma^2}\left(\textbf{u}^{\text{T}}\textbf{u}+\textbf{\text{U}}\right), 
\end{equation*}
where $\textbf{u}=(1,\bar{x}_1^*,\bar{x}_2^*,\ldots, \bar{x}_p^*)$, $\textbf{\text{U}}=\begin{bmatrix}
0 & 0 & 0 & \cdots & 0 \\
0 & s_{\textbf{x}_1^*}^2 & s_{\textbf{x}_1^*\textbf{x}_2^*} & \cdots & s_{\textbf{x}_1^*\textbf{x}_p^*}\\
0 & s_{\textbf{x}_1^*\textbf{x}_2^*} & s_{\textbf{x}_2^*}^2 & \cdots & s_{\textbf{x}_2^*\textbf{x}_p^*}\\
\vdots &	\vdots & \vdots & \ddots & \vdots \\
0 &	 s_{\textbf{x}_1^*\textbf{x}_p^*} & s_{\textbf{x}_2^*\textbf{x}_p^*} & \cdots & s_{\textbf{x}_p^*}^2
\end{bmatrix}$, $\bar{x}_j^*=\dfrac{\sum_{i=1}^{n}\delta_ix_{ij}}{k}$, $s_{\textbf{x}_j^*}^2=\dfrac{\sum_{i=1}^{n}\delta_i(x_{ij}-\bar{x}_j^*)^2}{k}$, $j=1,2,\ldots,p$, and $s_{\textbf{x}_o^*\textbf{x}_j^*}=\dfrac{\sum_{i=1}^{n}\delta_ix_{io}x_{ij}}{k}-\bar{x}_o^*\bar{x}_j^*$, $o\ne j=1,2,\ldots,p$.\\

Since $\text{det}(\textbf{\text{U}})=0$, that is $\textbf{\text{U}}$ is not invertible, we get that: 
\begin{equation*}
\text{det}\left(\textbf{Q}_{\text{Sub}}\right)=\dfrac{k^{p+1}}{\sigma^{2(p+1)}}\left(\text{det}(\textbf{\text{U}})+\textbf{u}\text{adj}(\textbf{\text{U}})\textbf{u}^{\text{T}}\right),
\end{equation*}
where \text{adj}(\textbf{U}) is the adjugate matrix of (\textbf{U}), or
\begin{equation}\label{genvar}
\text{det}\left(\textbf{Q}_{\text{Sub}}\right)=\dfrac{k^{p+1}}{\sigma^{2(p+1)}}\text{det}\left(\begin{bmatrix}
	s_{\textbf{x}_1^*}^2 & s_{\textbf{x}_1^*\textbf{x}_2^*} & \cdots & s_{\textbf{x}_1^*\textbf{x}_p^*}\\
	s_{\textbf{x}_1^*\textbf{x}_2^*} & s_{\textbf{x}_2^*}^2 & \cdots & s_{\textbf{x}_2^*\textbf{x}_p^*}\\
	\vdots & \vdots & \ddots & \vdots \\
	 s_{\textbf{x}_1^*\textbf{x}_p^*} & s_{\textbf{x}_2^*\textbf{x}_p^*} & \cdots & s_{\textbf{x}_p^*}^2
\end{bmatrix}\right).
\end{equation}

The expression \eqref{genvar} is the generalized variance \citep{wilks1932genvar} of covariates $\textbf{x}_j^*$, $j=1,2,\ldots,p$. Therefore, the problem of maximizing the determinant of the information matrix in \eqref{inf_sub} can be addressed as a problem of maximizing the generalized variance of covariates under the selected subdata.

Let $\textbf{A}=\begin{bmatrix}
	s_{\textbf{x}_1^*}^2 & s_{\textbf{x}_1^*\textbf{x}_2^*} & \cdots & s_{\textbf{x}_1^*\textbf{x}_p^*}\\
	s_{\textbf{x}_1^*\textbf{x}_2^*} & s_{\textbf{x}_2^*}^2 & \cdots & s_{\textbf{x}_2^*\textbf{x}_p^*}\\
	\vdots & \vdots & \ddots & \vdots \\
	 s_{\textbf{x}_1^*\textbf{x}_p^*} & s_{\textbf{x}_2^*\textbf{x}_p^*} & \cdots & s_{\textbf{x}_p^*}^2
\end{bmatrix}$. Applying Cholesky decomposition to $\textbf{A}=(A_{jo})$, $j,o=1,2,\ldots,p$, we get that:
\begin{equation*}
\textbf{A}=\textbf{L}\textbf{L}^{\text{T}},
\end{equation*}
where $\textbf{L}=(L_{jo})$, $j,o=1,2,\ldots,p$ is a real lower triangular matrix such that:
\begin{equation*}
L_{jj}=\sqrt{A_{jj}-\sum_{o=1}^{j-1}L_{jo}^2} \quad \text{and} \quad L_{jo}=\dfrac{A_{jo}-\sum_{h=1}^{o-1}L_{jh}L_{oh}}{L_{oo}}, j>o.
\end{equation*}
Therefore, we get that:
\begin{equation}\label{genvar_chol}
\text{det}(\textbf{A})=\prod_{j=1}^{p}L_{jj}^2.
\end{equation}

\begin{theorem}\label{theorem}
The generalized variance of covariates under the subdata is maximized by the selection of data points for which $s_{\textbf{x}_j^*}^2$ is maximized for any $j=1,2,\ldots,p$, and $s_{\textbf{x}_o^*\textbf{x}_j^*}=0$ for any $j>o=1,2,\ldots,j-1$, simultaneously.
\end{theorem}
\begin{proof}
According to \eqref{genvar_chol}, the generalized variance of covariates under the subdata is maximized when $L_{jj}^2$, or $A_{jj}-\sum_{o=1}^{j-1}L_{jo}^2$, are maximized for any $j=1,2,\ldots,p$. Therefore, the maximization of $A_{jj}$ for any $j=1,2,\ldots,p$ and $L_{jo}^2=0$, or $L_{jo}=0$, for any $o=1,2,\ldots,j-1$ are required, simultaneously.

For $j=1$ we get that $L_{11}^2=A_{11}$, and so the maximization of $s_{\textbf{x}_1^*}^2$ is derived. For $j=2$ we get that $L_{21}=A_{21}/L_{11}$, and so $A_{21}=0$ is required. Therefore, the maximization of $s_{\textbf{x}_2^*}^2$ and $s_{\textbf{x}_1^*\textbf{x}_2^*}=0$ are derived.

For any $j=3,4,\ldots,p$ we get that $L_{j1}=A_{j1}/L_{11}$ and so $A_{j1}=0$ is required. Also, for any $j=3,4,\ldots,p$ and for any $o=2,3,\ldots,j-1$ we get that $\sum_{h=1}^{o-1}L_{jh}L_{oh}=0$, since both $L_{jh}=0$ and $L_{oh}=0$ for any $h=1,2,\ldots,o-1$, and so $A_{jo}=0$ is required. Therefore, the maximization of $s_{\textbf{x}_j^*}^2$ and $s_{\textbf{x}_o^*\textbf{x}_j^*}=0$ for any $j=3,4,\ldots,p$ and for any $o=1,2,\ldots,j-1$ are derived.
\end{proof}
According to Theorem \ref{theorem}, the determinant of the information matrix in \eqref{inf_sub} is maximized selecting data points for which $s_{\textbf{x}_j^*}^2$'s are maximized  for any $j=1,2,\ldots,p$, and $s_{\textbf{x}_o^*\textbf{x}_j^*}=0$ for any $j>o=1,2,\ldots,j-1$, simultaneously. Even though such a case may not be feasible, in order to maximize the determinant of the information matrix in \eqref{inf_sub}, one should be interested in collecting data points for which $s_{\textbf{x}_j^*}^2$, $j=1,2,\ldots,p$ are getting as maximum as possible and $s_{\textbf{x}_o^*\textbf{x}_j^*}$, $j>o=1,2,\ldots,j-1$ tends to zero, simultaneously.

At first glance, it seems that the aforementioned theoretical result is already known in the current literature. However, such a statement is valid just in case $\textbf{x}_i$'s in model \eqref{model1} are factors not covariates, since it is well-known in the theory of experimental designs, that to prove the D-optimality of a design, the off-diagonal elements of the corresponding information matrix should be equal to zero, given that the sum of the diagonal elements is fixed. If the $\textbf{x}_i$'s in model \eqref{model1} are covariates, the sum of the diagonal elements of the information matrix is not fixed under any subdata, since even the change of only a data point leads to different values for all elements of the information matrix. As a result, working with covariates, one should take into consideration all the elements of the information matrix, and so we are interested in selecting data points for which $s_{\textbf{x}_j^*}^2$'s are maximized (as much as possible) for any $j=1,2,\ldots,p$, and $s_{\textbf{x}_o^*\textbf{x}_j^*}=0$ for any $j>o=1,2,\ldots,j-1$, simultaneously. Also, it is not obvious that the D-optimal subdata are those for which $s_{\textbf{x}_o^*\textbf{x}_j^*}=0$ for any $j>o=1,2,\ldots,j-1$. An option to maximize the determinant under the subdata, i.e., the generalized variance, is the selection of another data point that could lead to $s_{\textbf{x}_o^*\textbf{x}_j^*}=0$ for any $j>o=1,2,\ldots,j-1$ as in the previous subdata, but maximizing at least one of the previous $s_{\textbf{x}_j^*}^2$, $j=1,2,\ldots,p$, simultaneously. Another option is the selection of another data point that could not necessarily lead to $s_{\textbf{x}_o^*\textbf{x}_j^*}=0$ for any $j>o=1,2,\ldots,j-1$ as in the previous subdata, but maximizing at least one of the previous $s_{\textbf{x}_j^*}^2$, $j=1,2,\ldots,p$ to such an extent, that the new obtained determinant could be bigger than the one under the previous selected subdata.

\section{Existing approaches}\label{back}

We describe briefly existing approaches upon which we want to improve. 
\subsection{The IBOSS algorithm}\label{section_iboss}
\cite{wang2019information} proposed the IBOSS approach in order to select subdata. The scope of their approach is to obtain, given the size of the subdata, data points that provide most of the information contained in the full data. Their idea was motivated by the concept of optimal experimental designs, which aims at organizing, conducting, and interpreting results of experiments in an efficient manner in order to obtain as much useful information as possible, given a budget. More specifically, their idea was to apply the ``maximization'' of an information matrix, that takes place in optimal experimental designs, to recognize data points that provide most of the information contained in the full data. Therefore, they concluded that an optimal estimator of the unknown parameters of a linear regression model, based on the subdata, can be obtained by ``maximizing'' the inverse of the covariance matrix of the unknown parameters. They were interested in maximizing a matrix, and so the choice of an optimality criterion function of the matrix was needed. They preferred the popular D-optimality criterion that maximizes the determinant of the inverse of the covariance matrix of the unknown parameters given the subdata.

Furthermore, \cite{wang2019information} developed an algorithm that was based on their information-based optimal subdata selection framework. Their algorithm actually was motivated by their result for an upper bound of the determinant of the inverse of the covariance matrix of the unknown parameters given the subdata. The approach is based on collecting subdata with extreme covariate values, both small and large, occurring with the same frequency. Therefore, the algorithm of the IBOSS approach selects data points with the smallest as well as largest values of all covariates sequentially, given that previous selected data points are excluded. The time complexity of the algorithm of the IBOSS approach is $O(np+kp^2)$.

The IBOSS algorithm  outperforms the uniform and the leverage-based subsampling approaches in selecting informative subdata from big data. Also, the IBOSS algorithm compares favorably to the leverage-based subsampling approach in terms of the execution time, both being more efficient than modeling on the full data.

The basic approach in \cite{wang2019information} has found several extensions to other cases like
\cite{cheng2020information} for logistic regression, \cite{yao2019optimal} for multinomial logistic regression, \cite{wang2021optimal} for quantile regression. Algorithmic extensions can be found in
\cite{wang2019divide} and \cite{lee2021fast}.
For further work on the topic see \cite{yao2021review} and \cite{deldossi2021optimal}.

\subsection{The OSS algorithm}
\cite{wang2021oss}, motivated by \cite{wang2019information}, proposed a method to select subdata, with a focus on linear regression models, based on an orthogonal array (OA). Their purpose was to obtain subdata that approach an OA. The approach is driven by the fact that a two-level OA represents an optimal design for linear regression, since it minimizes the average variance of the estimated parameters as well as provides the best predictions \citep{dey&mukerjee}. An OA represents an optimal design because of its combinatorial orthogonality, and so the OSS approach selects data points with maximum combinatorial orthogonality. 

 Based on the combinatorial orthogonality of an OA, \cite{wang2021oss} developed a sequential addition algorithm that selects data points from the full data focusing on the maximization of the subdata's orthogonality. Also, some data points are deleted in each step of the algorithm in order to reduce the number of candidate data points, and so to speed up the algorithm. All covariates are scaled to [$-1$, $1$]. 
The OSS approach, and so the developed algorithm, is based on a discrepancy function that was defined by \cite{wang2021oss} in order to measure the distortion of data points on keeping two features simultaneously. The two features are connected with the optimality of orthogonal arrays. The first feature is that data points are located at the corners of the data domain $[-1, 1]^p$ and have large distances from the center. Thus, extreme data points provide more information about the model. The second feature is the combinatorial orthogonality, that is data points (their signs) are as dissimilar as possible. The computational complexity of the algorithm of the OSS approach is $O(np\text{log}k)$.

A design is A-optimal when the trace of the inverse of its information matrix is minimized in the class of all designs. \cite{wang2021oss} evaluated their approach by calculating, among others, the D- and A-efficiencies of the selected subdata, since a two-level OA represents a D-, A-optimal design for linear regression. As in \cite{wang2021oss}, the D- and A-efficiencies of the subdata can be calculated using:
$$
\text{D}_{eff}=\dfrac{\text{det}(\textbf{Q}_{\text{Sub}})^{1/(p+1)}}{k}
$$
and
$$
\text{A}_{eff}=\dfrac{p+1}{k\sum_{j=1}^{p+1}\lambda_j(\textbf{Q}_{\text{Sub}}^{-1})},
$$
respectively, where $\lambda_j(\textbf{Q}_{\text{Sub}}^{-1})$ denotes the $j$th eigenvalue of $\textbf{Q}_{\text{Sub}}^{-1}$.

The D-efficiency of the subdata can be interpreted as a measure of how close the selected subdata are to the best possible subdata, for which the determinant of the corresponding information matrix is maximized. Similarly, the A-efficiency of the subdata quantifies how close the selected subdata are to the best possible subdata for which the trace of the inverse of their corresponding information matrix is minimized. For example, subdata with 100\% D-efficiency indicates that the subdata are as efficient as the D-optimal ones. Therefore, subdata with higher D- and A-efficiencies are preferable.

The OSS approach outperforms the uniform subsampling and the IBOSS approach in selecting informative subdata from big data. Also, the OSS approach is faster than the IBOSS one, and they are both faster than modeling on the full data. It is important to mention that the IBOSS approach selects data points with only the first aforementioned feature without taking into any consideration of the second feature. Therefore, the consideration of the second feature is an important improvement of the OSS approach compared with the IBOSS one.

 Theorem \ref{theorem} shows that the perfect case is achieved under orthogonality taking the extreme points. The OSS algorithm, since the covariates are
 continuous, cannot achieve this.  Since it is too difficult in practice to achieve zero covariance between the covariates, Theorem \ref{theorem}  implies that a good strategy in practice towards a better subdata selection is to select observations that maximize the generalized variance.

\subsection{The approach of \cite{ren&zhao}}
\cite{ren&zhao}, motivated by \cite{wang2021oss}, defined
a new discrepancy function to evaluate subdata, based on a two-dimensional projection property of orthogonal arrays, that the approach of \cite{wang2021oss} did not take into account. Also, they developed three algorithms, that are suitable for datasets with a small number of covariates, in order to select subdata according to the distribution of covariates as well as the prior information concerning the importance ordering of covariates.

According to the OSS approach, the dissimilar data points are selected
from all covariates of the full data. The basic idea of the approach of \cite{ren&zhao} is to select dissimilar data points from any two covariates of the full data, so as the selected subdata to approach a $k\times p$ two-level OA of strength 2, OA($k$, $p$, $2$, $2$), that is a $k\times p$ matrix whose columns have the property that in every pair of columns each of the possible ordered pairs of elements appears the same number of times. Therefore, the approach of \cite{ren&zhao} attempts to improve both $D_{eff}$ and $A_{eff}$ of the selected subdata taking into consideration a projection property of an OA($k$, $p$, $2$, $2$), that is the combinations ($1$, $1$), ($1$, $-1$), ($-1$, $1$), ($-1$, $-1$) appear $k/4$ times in any two of its columns. Note that an OA($k$, $p$, $2$, $2$) is D-, A-optimal in the class of all two-level fractional factorial designs for linear regressions \citep{dey&mukerjee}.

In the approach of \cite{ren&zhao}, following the OSS approach, all covariates are scaled to [$-1$, $1$]. According to the aforementioned two-dimensional projection property of an OA($k$, $p$, $2$, $2$), the four different sign combinations appear as equal as possible, ideally $k/4$ times each, in any two covariates of the selected subdata. Therefore, the new discrepancy function, defined by \cite{ren&zhao}, was based on this feature. 

The developed algorithms in the approach of \cite{ren&zhao} both select and eliminate data points in order to be as fast as possible. \cite{ren&zhao}, after conducting some simulated experiments in order to evaluate their developed algorithms, they concluded that their first algorithm is suitable in the case that there is prior information
concerning the importance ordering of covariates regardless their distribution. Their Algorithm 2 is recommended if the distribution of covariates is a heavy-tailed distribution, whereas their Algorithm 3 is recommended if the distribution of covariates is a normal distribution.

Since in Section \ref{section_simulation} we work on covariates that have a multivariate normal distribution, we provide some concerns about Algorithm 3 in the approach of \cite{ren&zhao}. Algorithm 3, in Step 3, either deletes data points from the subdata adding them into a candidate set of data points, or deletes data points from a candidate set of data points adding them into the subdata. The case that Algorithm 3 deletes data points from the subdata adding them into a candidate set of data points works perfectly, since Algorithm 3 is always able to maintain a number of data points to the subdata. The remaining data points are deleted and they are added into a candidate set of data points. The case that Algorithm 3 deletes data points from a candidate set of data points adding them into the subdata is not always implementable, since data points that are going to be added into the subdata may not exist in the candidate set of data points in terms of quantity. This issue leads to the non-implementation of Algorithm 3. Since \cite{ren&zhao} do not provide any advice in this situation, we do not compare the approach of \cite{ren&zhao} with our approach. However, we need to mention that the approach of \cite{ren&zhao} is in the right direction, since it takes into consideration the two-dimensional projection property of an OA($k$, $p$, $2$, $2$).

According to \cite{ren&zhao}, Algorithm 3 outperforms both the IBOSS and the OSS approaches in selecting informative subdata from big data. The time complexity of Algorithm 3 of the approach of \cite{ren&zhao} is $\text{max}\left(O(np),O(ftp^2)\right)$, where $f$ is the $f$th iteration of Algorithm 3 and $t$ is the number of data points of the candidate set.

\section{The new proposed algorithms}\label{section_algorithm}
\subsection{The Alg1}
We develop an algorithm with a primary focus on improving the algorithm of an already existing approach, such that a reselection of subdata leads to an increase of the generalized variance. 

The input of our algorithm is the subdata, say set $\textbf{S}=(\textbf{s}_{i}), i=1,2,\ldots,k$,  obtained by some approach, i.e., the OSS or the IBOSS approach, and so the implementation of our algorithm requires the implementation of the corresponding algorithm first. Our goal is to identify and interchange selected data points by the corresponding algorithm  with those that were not selected, say set $\textbf{D}=(\textbf{d}_{r\cdot}), r=1,2,\ldots,n-k$. The criterion that allows the aforementioned interchange is the increase of the generalized variance, denoted as $V$. At first, we need to mention that we do not take into consideration all data points that were not selected by the corresponding algorithm (set \textbf{D}) as candidate data points of the final subdata. Therefore, a set of candidate points should be obtained. Such a set, denoted by \textbf{F}, is a subset of data points among the ones that have not been already selected by the corresponding algorithm, that is $\textbf{F}\subset\textbf{D}$. In order to provide a more comprehensive explanation about \textbf{F}, we need to clarify how \textbf{F} is obtained. Data points in \textbf{F} are selected in a manner similar to how the algorithm of the IBOSS approach selects data points, with a notable difference. We reiterate (see Section \ref{section_iboss}) that the algorithm of the IBOSS approach selects data points with the smallest as well as largest values of all covariates sequentially, given that previous selected data points are excluded. So, in order to obtain \textbf{F}, we select data points from \textbf{D} with the smallest as well as largest values of all covariates sequentially, but previous selected data points are not excluded. This is the difference in the selection of data points between the algorithm of the IBOSS approach and the construction of \textbf{F}. A consequence of such a construction of \textbf{F} is that a data point could be selected more than once. This is a circumstance that may occur when a data point is selected to belong to \textbf{F} due to its values in multiple covariates. To clarify and without loss of generality, a data point in \textbf{D} can be selected twice for the construction of \textbf{F}, if its value in the first covariate is among the smallest ones and its value in the second covariate is among the largest ones. This situation may arise because the selected data point is not excluded from \textbf{D} when we select data points based on the values of the first covariate. Therefore, the already selected data point still remains in \textbf{D}, and so it can be selected again when we select data points based on the values of the second covariate to obtain \textbf{F}. Thus, such a method of selecting data points can lead to at least duplicated ones in \textbf{F}, and so only one of them is kept, that is we keep only unique data points in \textbf{F}. Also, we are not able to know the final size of the data points in \textbf{F}, say $N_{\text{F}}$, since it is not feasible to know in advance the existence of at least duplicated data points. However, the maximum final size of the data points in \textbf{F} is equal to $Kp$, where $K$ is an even number of data points selected for each covariate. Note that the value of $K$ is user-selected.   

We will provide some explanations about the second step of Alg1, in order to clarify its functionality. The selection of data points from \textbf{D} with the smallest as well as largest values of all covariates sequentially, requires to sort the values for each of the $p$ covariates in ascending (or descending) order. The notation for this procedure in Alg1 is $\text{sort}(\textbf{d}_{\cdot j})$, $j=1,2,\ldots,p$. Also, if only one covariate is sorted within \textbf{D}, without sorting all the remaining covariates based on the selected covariate that has been sorted, then the resulting \textbf{D} will consist of different data points. Therefore, all the remaining covariates in \textbf{D} should be sorted according to the covariate that has been sorted. The notation for this procedure in Alg1 is sort(\textbf{D}).

\begin{algorithm}
	\caption{Alg1}
	\begin{algorithmic}
		\Require subdata $\textbf{S}=(\textbf{s}_{i}), i=1,2,\ldots,k$  of some approach, full data $\textbf{D}_{\text{Full}}$, subdata size $k$, candidate data points $K$ from each covariate
		\Ensure new obtained subdata \textbf{S}
		\State \textbf{Step 1: Preparation}
		\State $V=\text{det}\left(\textbf{Q}_{\text{Sub}}\right)$ \Comment{generalized variance of $\textbf{S}$}
		\State $\textbf{D}=\textbf{D}_{\text{Full}}-\textbf{S}=(\textbf{d}_{r\cdot})$ \Comment{remaining data points $\textbf{d}_{r\cdot}=(d_{r1},\ldots,d_{rp})\not\in \textbf{S}$}
		\State $N_{\text{D}}=\text{nrow}(\textbf{D})$ \Comment{number of data points $\textbf{d}_{r\cdot}\in \textbf{D}$}
		\State $\textbf{F}=\O$ \Comment{initialize the index set of candidate data points}
		\State \textbf{Step 2: Find candidate data points}
		\For{$j$ in $1,\ldots,p$}
		\State $\textbf{d}_{\cdot j}=\text{sort}(\textbf{d}_{\cdot j})$ \Comment{sort $\textbf{d}_{\cdot j}=(d_{1j},\ldots,d_{N_{\text{D}}j})$}
		\State $\textbf{D}=\text{sort}(\textbf{D})$ \Comment{sort \textbf{D} based on $\textbf{d}_{\cdot j}$}
		\State $\textbf{F}=\textbf{F}\cup\{\textbf{d}_{1\cdot}\}\cup\cdots\cup\{\textbf{d}_{K/2\cdot}\}$
		\State $\textbf{F}=\textbf{F}\cup\{\textbf{d}_{N_{\text{D}}-K/2+1\cdot}\}\cup\cdots\cup \{\textbf{d}_{N_{\text{D}}\cdot}\}$
		\EndFor
		\State $\textbf{F}=\text{unique}(\textbf{F})$ \Comment{keep unique data points of $\textbf{F}=(\textbf{f}_w)$}
		\State $N_{\text{F}}=\text{nrow}(\textbf{F})$ \Comment{number of data points $\textbf{f}_{w}\in \textbf{F}$}
		\State \textbf{Step 3: Main algorithm}
		\For{$i$ in $1,\ldots,k$}
		\For{$w$ in $1,\ldots,N_{\text{F}}$}
		\State $\textbf{s}_i\leftrightarrow\textbf{f}_w$ \Comment{interchange data points $\textbf{s}_i$ and $\textbf{f}_w$}
		\State $V_{\text{new}}=\text{det}\left(\textbf{Q}_{\text{Sub}}\right)$ \Comment{generalized variance of new \textbf{S}} 
		\If{$V_\text{new}>V$}
		\State $V=V_\text{new}$
		\State \textbf{break}
		\Else
		\State $\textbf{s}_i\leftrightarrow\textbf{f}_w$
		\EndIf
		\EndFor
		\EndFor
		\State \textbf{return} \textbf{S}
	\end{algorithmic}
\end{algorithm}

\begin{remark}\label{rem_alg}
In Alg1, the data points $\textbf{f}_w$, $w=1,2,\ldots,N_{\text{F}}$ that are interchanged with the data points $\textbf{s}_i$, $i=1,2,\ldots,k$, are selected in the order in which $\textbf{F}=(\textbf{f}_w)$ is constructed in Step 2.
\end{remark}

A consequence of Remark \ref{rem_alg} is that the interchange of a data point $\textbf{s}_i$ with a later data point of $\textbf{F}$ could possibly lead to a greater generalized variance $V_{\text{new}}$. However, it is not feasible to determine if such a case can occur, since Alg1 interchanges a data point $\textbf{s}_i$ with the first data point that locates in $\textbf{F}$ given that $V_\text{new}>V$.

The time complexity to construct $\textbf{F}$ is $O((n-k)p)$. The procedure of interchanging data points between $\textbf{S}$ and $\textbf{F}$ has time complexity $O(kN_{\text{F}})$. Thus, the time complexity of Alg1 is $O(np\text{log}k+(n-k)p+kN_{\text{F}})$.  In the remaining of this paper, in the cases that Alg1 is executed with more than one iteration, only Step 3 is repeated. 

\subsection{The VAlg1}
The aforementioned consequence of Remark \ref{rem_alg} proposes a variation on Alg1, and so a new algorithm (VAlg1) seeks to improve the generalized variance of the final subdata. VAlg1 makes a change as to which data point of $\textbf{F}$ is chosen to be interchanged with a data point of $\textbf{S}$. Therefore, a data point $\textbf{s}_i$ is interchanged with that data point of $\textbf{F}$, among all remaining ones, for which $V_\text{new}$ is the largest possible. The interchanging between data points in VAlg1 is applied to all data points of $\textbf{S}$.

\begin{algorithm}[H]
	\caption{VAlg1}
	\begin{algorithmic}
		\State \textbf{Steps 1 and 2: Same as in Alg1}
        \State \textbf{Step 3: Main algorithm}
		\For{$i$ in $1,\ldots,k$}
		\For{$w$ in $1,\ldots,N_{\text{F}}$}
		\State $\textbf{s}_i\leftrightarrow\textbf{f}_w$ \Comment{interchange data points $\textbf{s}_i$ and $\textbf{f}_w$}
		\State $V_{\text{new}}=\text{det}\left(\textbf{Q}_{\text{Sub}}\right)$ \Comment{generalized variance of new \textbf{S}} 
		\If{$V_\text{new}>V$}
		\State $V=V_\text{new}$
		\Else
		\State $\textbf{s}_i\leftrightarrow\textbf{f}_w$
		\EndIf
		\EndFor
		\EndFor
		\State \textbf{return} \textbf{S}
	\end{algorithmic}
\end{algorithm}

VAlg1 has the same time complexity as Alg1. More details about the execution times of Alg1 as well as VAlg1 are provided in Section \ref{section_simulation}.

Note that both Alg1 and VAlg1 are, in fact, an  attempt to implement what Theorem \ref{theorem} indicates. Attempting to increase the variance of  covariates under the subdata we collect as candidate data points, data points close to the extreme ones, i.e., data points that are considered as more probable to lead us close  to orthogonality.

\subsection{Dealing with outliers}\label{deal_outliers}
In big data analytics, addressing the presence of outliers is an important challenge that should be carefully considered. The detection and handling of outliers need particular attention as both the complexity and volume of big data increase. 

In \cite{wang_outliers} the problems posed as well as the solutions proposed for outlier detection on big data are examined. More recently, \cite{deldossi2023} considered the presence of outliers in big data with the aim of selecting subdata for precise parameter estimation and accurate predictions in case of a linear regression model. Their motivation was that D-optimal support points are located on the boundary of the design space, and so a subdata selection based on the D-optimality criterion may result in the inclusion of outliers.

\cite{wang2019information} noted that the selected subdata based on the algorithm of the IBOSS approach may include outliers. They suggested that outliers in the subdata could be identified using outlier diagnostic methods, as the subdata follow the same underlying regression model as the full data. They also mentioned that if a data point is an outlier in the full data, then this data point will be identified as an outlier in the subdata as well. Moreover, they asserted that data points that are far from others should be used for parameter estimation, if they follow the underlying model.

\section{Simulation experiments}\label{section_simulation}
In this section, we use simulated data in order to evaluate the performance of both Alg1 and VAlg1. Moreover, we present the results of the algorithms of the approaches of IBOSS and OSS, for comparing them with our approach.

\subsection{Evaluation of algorithms}\label{simulation}
In this simulation experiment, observations $\textbf{x}_i$'s follow a multivariate normal distribution, that is, $\textbf{x}_i\sim N(\textbf{0},\mathbf{\Sigma})$, where $\mathbf{\Sigma}=\left(\Sigma_{ij}\right)$, $i,j=1,2,\ldots,p$ is a covariance matrix. Also, $\Sigma_{ij}=1$ for $i=j=1,2,\ldots,p$ and $\Sigma_{ij}=0.5$ for $i\ne j=1,2,\ldots,p$.
The response data are generated from the linear model in \eqref{model1} with the true value of $\boldsymbol{\beta}$ being a $11$ dimensional vector with all elements equal to $1$ and $\sigma^2=3$. An intercept is included, so $p=10$.

The simulation is repeated $1,000$ times. We calculate the D-, A-efficiencies and the mean squared error (MSE) of the subdata selected by our approach as well as the approaches of IBOSS and OSS. As shown in \cite{wang2019information} and \cite{wang2021oss}, we estimate the intercept with the adjusted estimator $\hat{\beta}_0=\bar{y}-\bar{\textbf{x}}^{\text{T}}\hat{\boldsymbol{\beta}}_1^{Sub}$, where $\bar{y}$ is the mean of the response full data, $\bar{\textbf{x}}$ is the vector of means of all covariates in the full data, and $\hat{\boldsymbol{\beta}}_1^{Sub}$ is the ordinary least squares estimate of $\boldsymbol{\beta}_1^{Sub}$ based on the subdata. Therefore, we consider $(\hat{\beta}_0^{(r)}-\beta_0)^2$ and $||\hat{\boldsymbol{\beta}}_1^{(r)}-\boldsymbol{\beta}_1||^2$ the MSE for intercept and slope estimators in the $r$th repetition, where $\hat{\beta}_0^{(r)}$ and $\hat{\boldsymbol{\beta}}_1^{(r)}$ are $\hat{\beta}_0$ and $\hat{\boldsymbol{\beta}}_1^{Sub}$ in the $r$th repetition.

We are interested in investigating several cases when the full data sizes are $n=5\times10^3, 10^4$ and $10^5$, and the subdata size is fixed at $k=100$. Also, the maximum final size of the candidate data points is fixed at $Kp=250$, that is $K=25$. Alg1 is executed with $5$ iterations, and VAlg1 is executed once. The choice of the number of iterations for Alg1 and VAlg1 is explained in Sections \ref{iterations} and \ref{time}, respectively. 

Figure \ref{fig1} shows the MSEs, D-efficiency, and A-efficiency for the subdata selected by different approaches. The mean values ($\blacklozenge$) are also provided.
\begin{figure}[!thb]
\begin{center}
\includegraphics[width=1\textwidth]{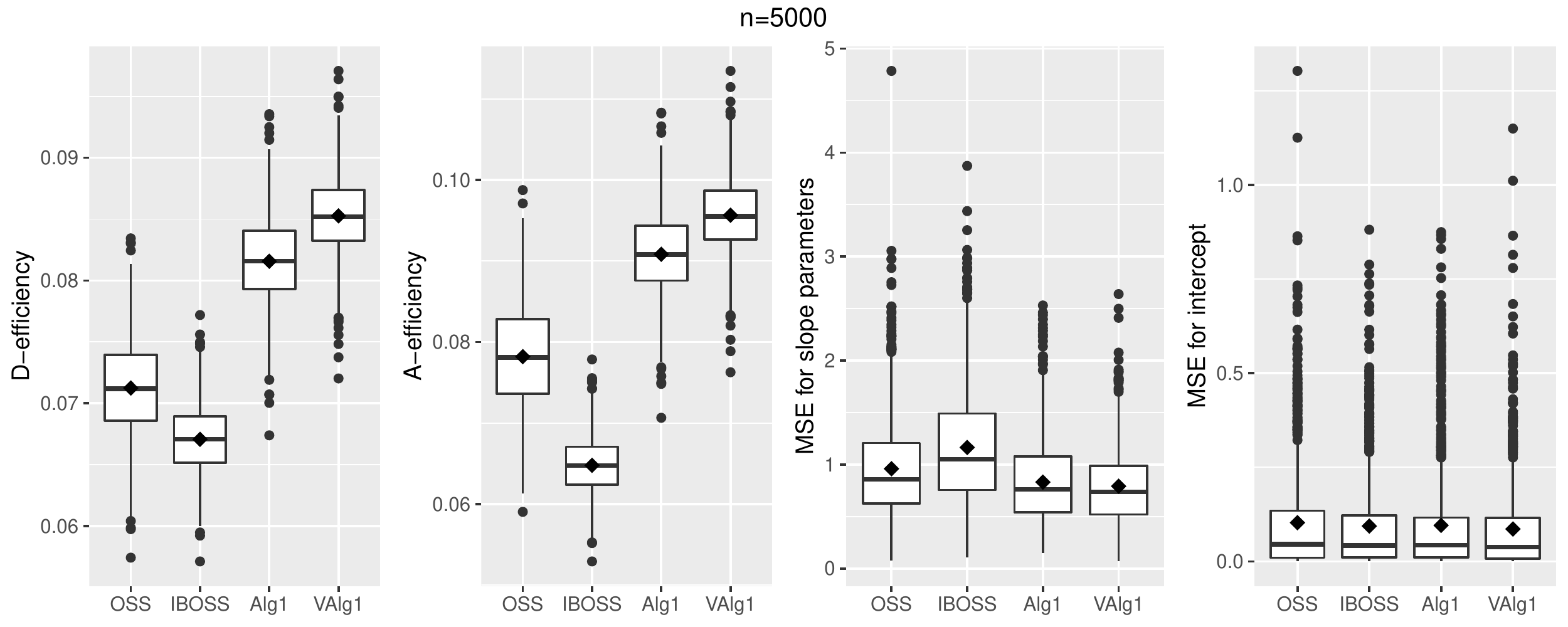}\par
\end{center}
\begin{center}
\includegraphics[width=1\textwidth]{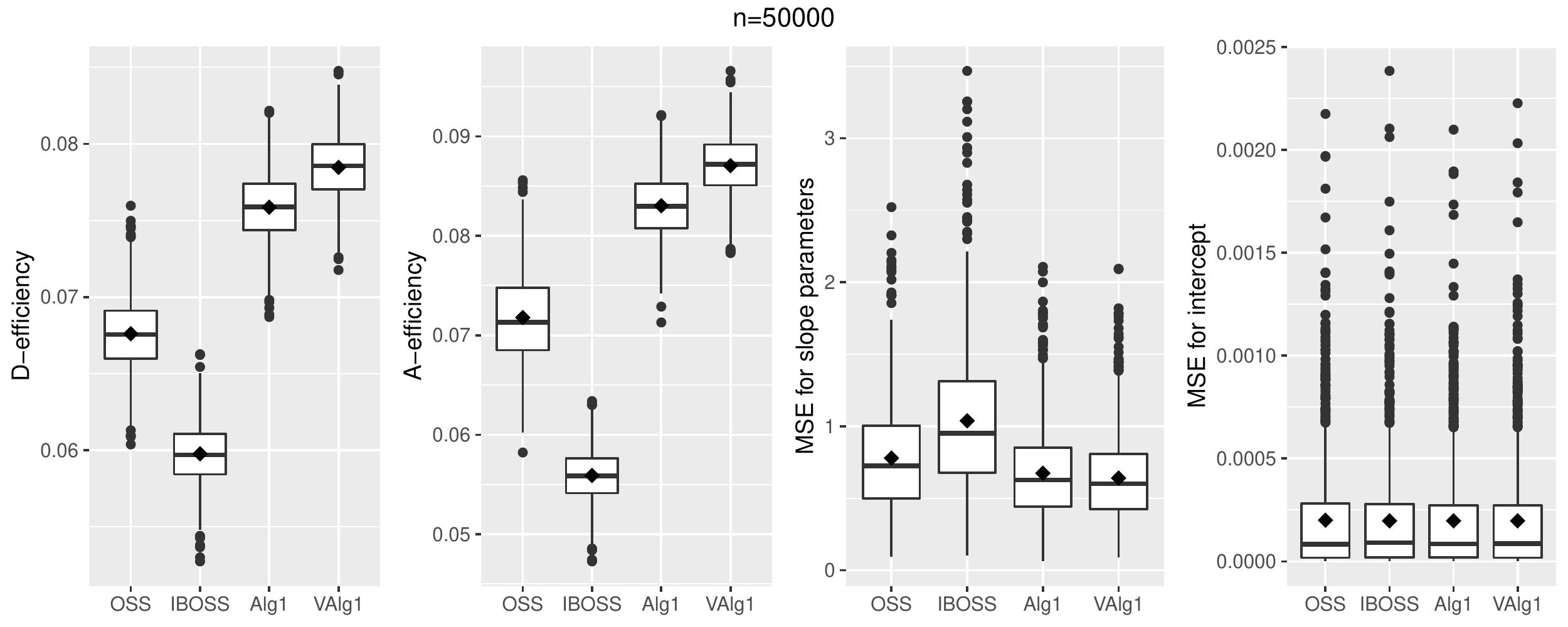}\par
\end{center}
\begin{center}
\includegraphics[width=1\textwidth]{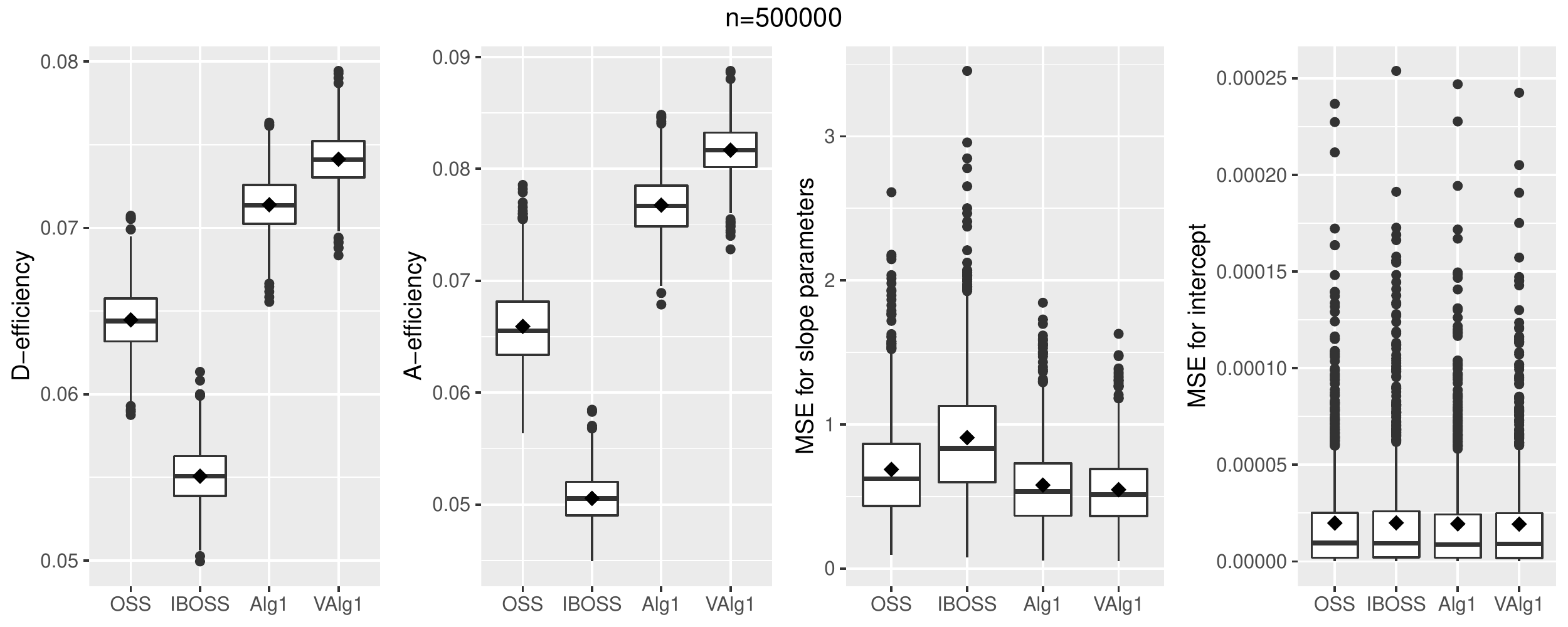}\par
\end{center}
\caption{The MSEs, D- and A-efficiencies for the subdata selected by different approaches, when the full data size is $n=5\times10^3, 10^4$ and $10^5$, the subdata size is $k=100$, and the number of candidate data points is $K=25$. Alg1 is executed with $5$ iterations, and VAlg1 is executed once.}
\label{fig1}
\end{figure}

Both Alg1 and VAlg1 outperform the IBOSS and OSS algorithms in each of the D-, A-optimality criteria, and provide more accurate estimates for the model parameters as one can see in Figure \ref{fig1}.  The important finding is the magnitude of progress in the D-, A-optimality criteria rather the absolute improvement which is sure since we start from the OSS algorithm and we proceed by improving the generalized variance.
As in the OSS approach, MSE for slope parameters from our approach decreases as the full data size $n$ increases, even though the subdata size is fixed at $k=100$, and so this fact indicates that our approach identifies more informative data points from the full data when the full data size increases. MSE for intercept is immutable among the three approaches. Also, VAlg1 can be considered to be more effective than Alg1. Such a result can find justification in the interchanging of data points between $\textbf{S}$ and $\textbf{F}$. One could argue that our approach is more effective when each data point of $\textbf{S}$ is interchanged with a data point of $\textbf{F}$ for which $V_{\text{new}}$ is the largest possible. However, Alg1 may be more effective than VAlg1 under other circumstances. Such a result could be achieved either executing Alg1 with more than $5$ iterations or applying it under another set of $n$, $p$, $k$ and $K$.

It is worth mentioning that our approach performs well under the A-optimality criterion, even though both Alg1 and VAlg1 are developed based on the increasing of the determinant of the information matrix of the subdata. Therefore, our approach could provide some insights on how to obtain subdata that approach an OA. 

\subsection{\label{time} Execution time of algorithms}
In this experiment, observations $\textbf{x}_i\sim N(\textbf{0},\mathbf{\Sigma})$, where $\mathbf{\Sigma}=\left(\Sigma_{ij}\right)$, $i,j=1,2,\ldots,p$ such that $\Sigma_{ij}=1$ for $i=j=1,2,\ldots,p$ and $\Sigma_{ij}=0.5$ for $i\ne j=1,2,\ldots,p$. The simulation is repeated $500$ times with full data size $n=10^3$. We focus on the execution times of both Alg1 and VAlg1 under various cases. All
computations, for this section but also for Section \ref{section_data} are carried out on a PC with 3.6 GHz Intel 8-Core I7
processor and 16GB memory. We fix $p=7$ covariates. Also, we combine the following two features:
\begin{enumerate}
 \item The subdata size $k$ is equal to $28$, $42$ and $56$.
 \item The value of $K$ is equal to $20$, $40$ and $60$.
\end{enumerate} 
At first, we are interested in investigating the execution time of Alg1 as a result of the number of its iterations. Alg1 is executed once as well as with $3$, $5$, $7$, $10$, $12$, $15$, $18$ and $20$ iterations. In Figure \ref{fig2}, we present the mean execution time (in seconds) of Alg1. It is easily perceived that the mean execution time of Alg1 is getting slower as the number of $k$ and $K$ increases. However, an almost linear increase of the mean execution time on the number of iterations suggests that Alg1 is worth implementing.

\begin{figure}[!thb]
\begin{center}
\includegraphics[width=1\textwidth]{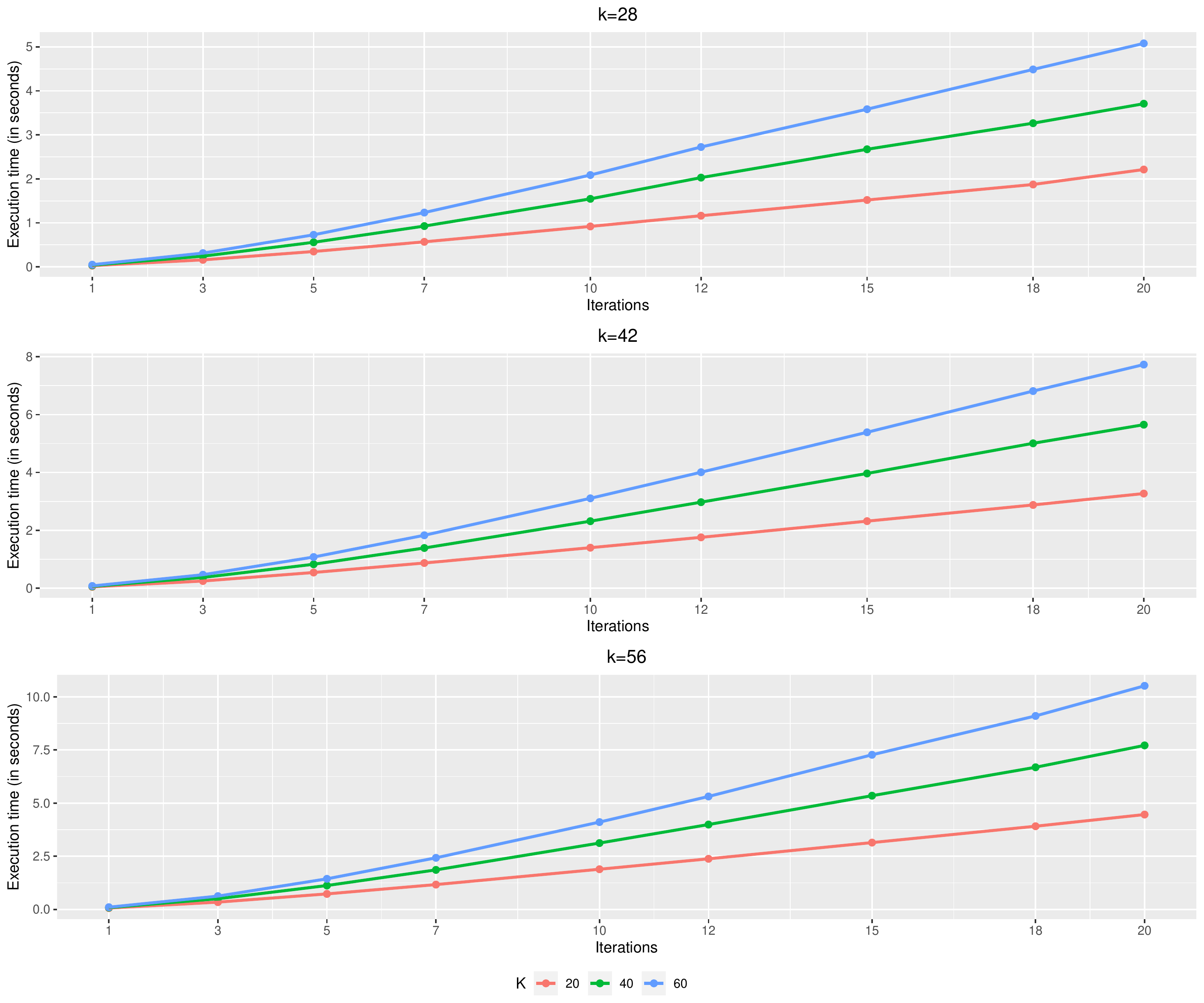}
\end{center}
\caption{The mean execution time (in seconds) of Alg1 in the cases of $1$, $3$, $5$, $7$, $10$, $12$, $15$, $18$ and $20$ iteration(s), when the subdata size is $k=28$, $42$, $56$, and the number of candidate data points is $K=20$, $40$, $60$.}
\label{fig2}
\end{figure}

In Table \ref{table1}, we present the mean execution time (in seconds) of VAlg1. VAlg1 is executed once, since it searches among all data points of $\textbf{F}$. The implementation of VAlg1 with more than one execution may result to a time-consuming algorithm without any improved results. As in Alg1, the mean execution time of VAlg1 is getting slower as the number of $k$ and $K$ increases.

\begin{table}[!htb]
\caption{The mean execution time (in seconds) of VAlg1 which is executed once, when  the subdata size is $k=28$, $42$, $56$, and  the number of candidate data points is  $K=20$, $40$, $60$.}
\label{table1}
\begin{center}
\begin{tabular}{cccccccccc}
\toprule
 $k$ & \multicolumn{3}{c}{28} & \multicolumn{3}{c}{42} & \multicolumn{3}{c}{56} \\ \midrule
 $K$ & 20 & 40 & 60 & 20 & 40 & 60 & 20 & 40 & 60 \\ \midrule
 \begin{tabular}{@{}c@{}}Mean execution \\ time (in seconds) \end{tabular} & 
0.106 & 0.214 & 0.307 & 0.197 & 0.317 & 0.434 & 0.230 & 0.451 & 0.600 \\ \bottomrule
\end{tabular}
\end{center}
\end{table}

 Both Alg1 and VAlg1 start from the OSS approach. Even though the execution times of both Alg1 and VAlg1 are added to the execution time of the OSS approach in order to obtain an optimal subdata, the results in Figure \ref{fig1} and Table \ref{table1} indicate that our approach is worth implementing. We need to mention that the number of $n$, $k$, $K$ as well as the number of iterations for Alg1 have been selected in order to provide just an illustration of our approach.
 
\subsection{\label{iterations} About iterations of Alg1}
 As we have mentioned, Alg1 can be executed for a number of iterations. This number can be large if we wish at the cost of additional computing time.  As one can imagine, the few starting iterations will lead to a large increase of the generalized variance which will start to slow down as the number of iterations increases.  Figure \ref{figiter} shows the relative mean percent increase in the generalized variance by Alg1 for a given number of iterations,  given that Alg1 starts from the OSS approach. The configuration is exactly the same as in Section \ref{time}. One can see that one iteration provides a huge increase in the generalized variance. The increase in the generalized variance is quite large up to $5$ iterations, and then it slows down. Taking into account the time performance in Section \ref{time}, we believe that a choice of $5$ iterations makes sense since it is compromise between a large increase and a doable computing time. Of course, in practice and if there are not significant time restrictions one may use more iterations. Figure \ref{figiter} shows that typically after 12 iterations of Alg1, the generalized variance improves quite a little. In conclusion, we suggest Alg1 to be executed with $5$ iterations, and VAlg1 (see Section \ref{time}) to be executed once.
  
 \begin{figure}[!thb]
\begin{center}
\includegraphics[width=1\textwidth]{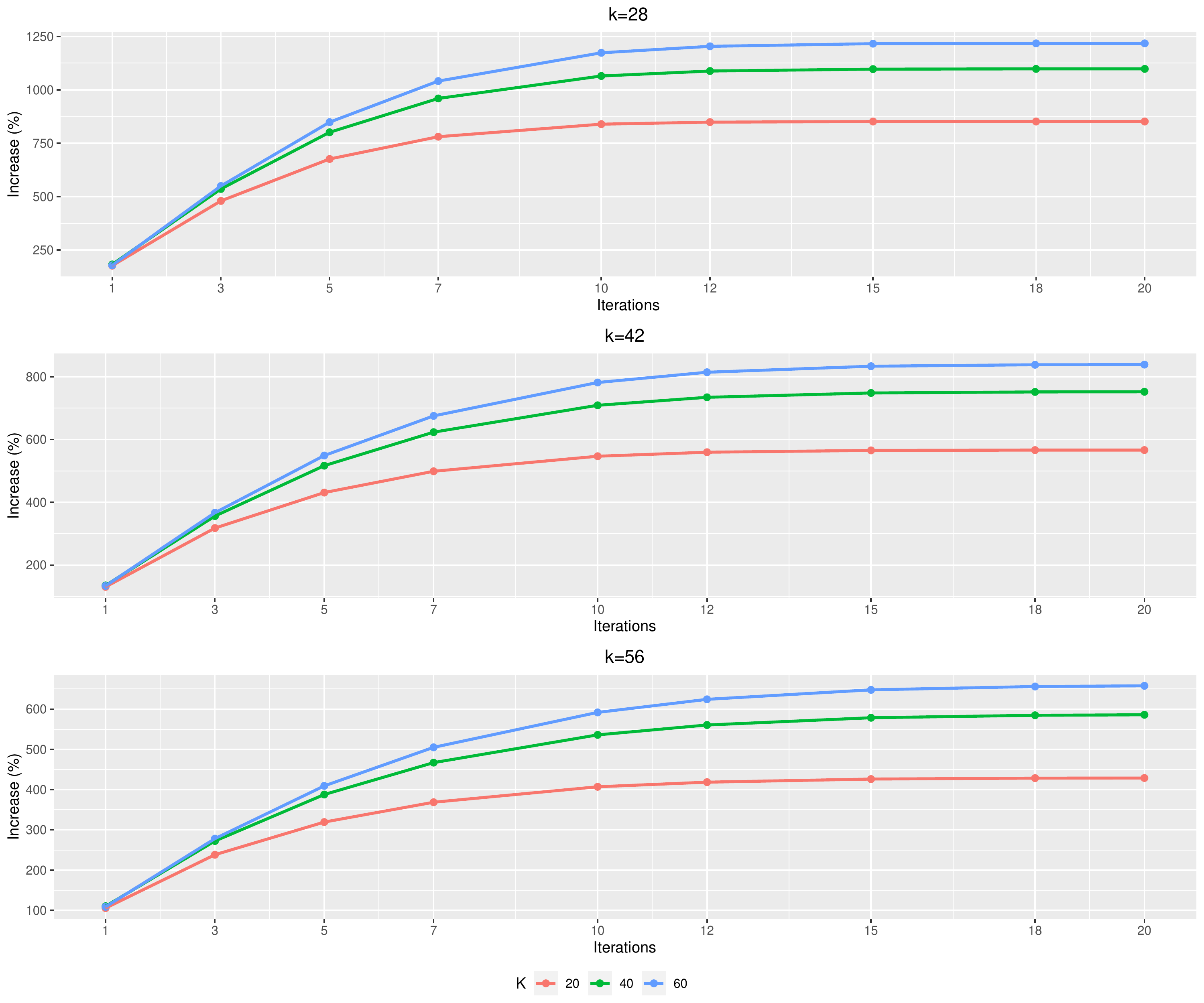}
\end{center}
\caption{The mean percent increase in the generalized variance by Alg1 in the cases of $1$, $3$, $5$, $7$, $10$, $12$, $15$, $18$ and $20$ iteration(s), when the subdata size is  $k=28$, $42$, $56$, and the number of candidate data points is $K=20$, $40$, $60$.}
\label{figiter}
\end{figure}

\subsection{\label{simulation_outliers} Evaluation of algorithms in presence of outliers}
In this simulation experiment, the full data size is $n=500,000$. The first $499,950$ observations $\textbf{x}_i\sim N(\textbf{0},\mathbf{\Sigma})$, and the remaining $50$ observations $\textbf{x}_i\sim N(\boldsymbol{\mu},\mathbf{\Sigma})$ with two scenarios:
 \begin{enumerate}
     \item $\boldsymbol{\mu}=\left(5,\textbf{0}_{p-1}\right)$
     \item $\boldsymbol{\mu}=\left(7,\textbf{0}_{p-1}\right)$ 
 \end{enumerate}
 The vector $\textbf{0}_{p-1}$ is a $(p-1)$-dimensional vector consisting of zeros. Also, $\mathbf{\Sigma}=\left(\Sigma_{ij}\right)$, $i,j=1,2,\ldots,p$ such that $\Sigma_{ij}=1$ for $i=j=1,2,\ldots,p$ and $\Sigma_{ij}=0.5$ for $i\ne j=1,2,\ldots,p$. The response data are generated from the linear model in \eqref{model1} with the true value of $\boldsymbol{\beta}$ being a $11$ dimensional vector with all elements equal to $1$ and $\sigma^2=3$. An intercept is included, so $p=10$.

The simulation is repeated $1000$ times. We calculate the MSE of the subdata selected by our approach as well as the approaches of IBOSS and OSS as in Section \ref{simulation}. The subdata size is fixed at $k=100$. Also, the maximum final size of the candidate data points is fixed at $Kp=250$, that is $K=25$. Alg1 is executed with $5$ iterations, and VAlg1 is executed once.

Figure \ref{fig1} shows the MSEs, D-efficiency, and A-efficiency for the subdata selected by different approaches. The mean values ($\blacklozenge$) are also provided.
\begin{figure}[!thb]
\begin{center}
\includegraphics[width=1\textwidth]{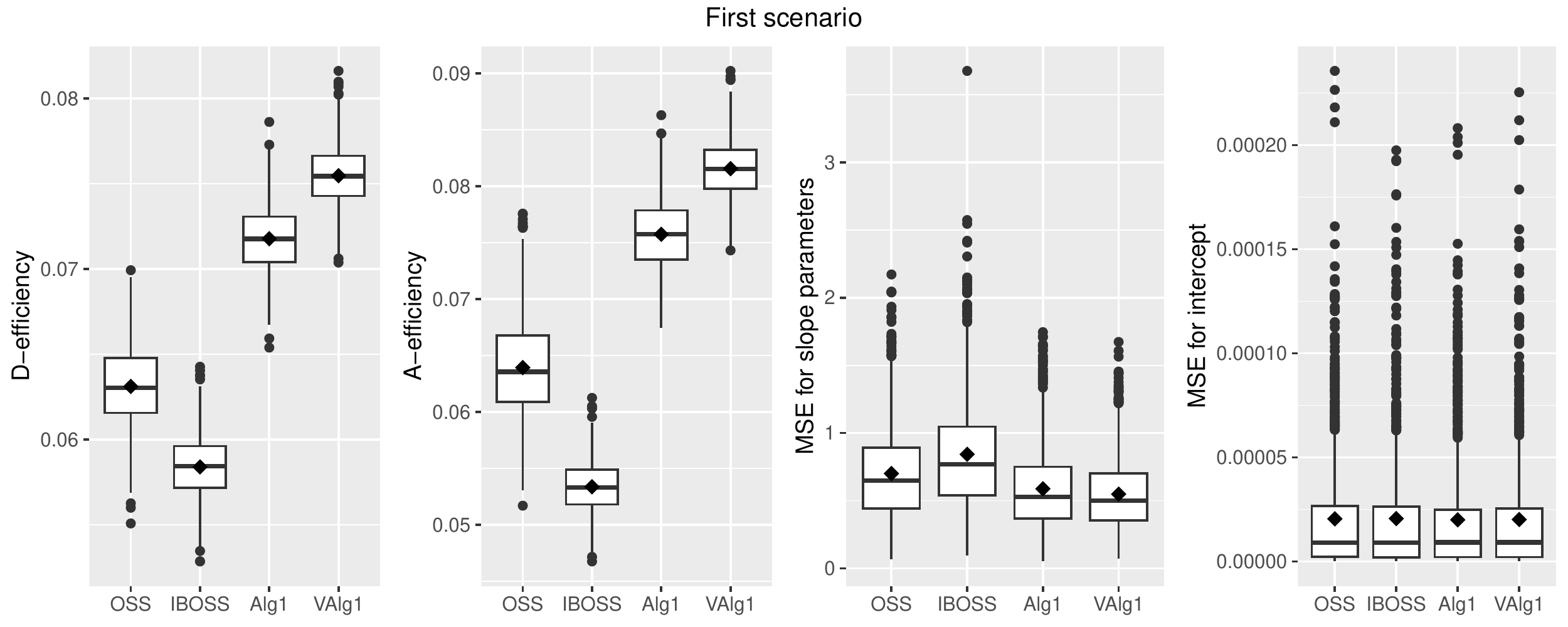}\par
\end{center}
\begin{center}
\includegraphics[width=1\textwidth]{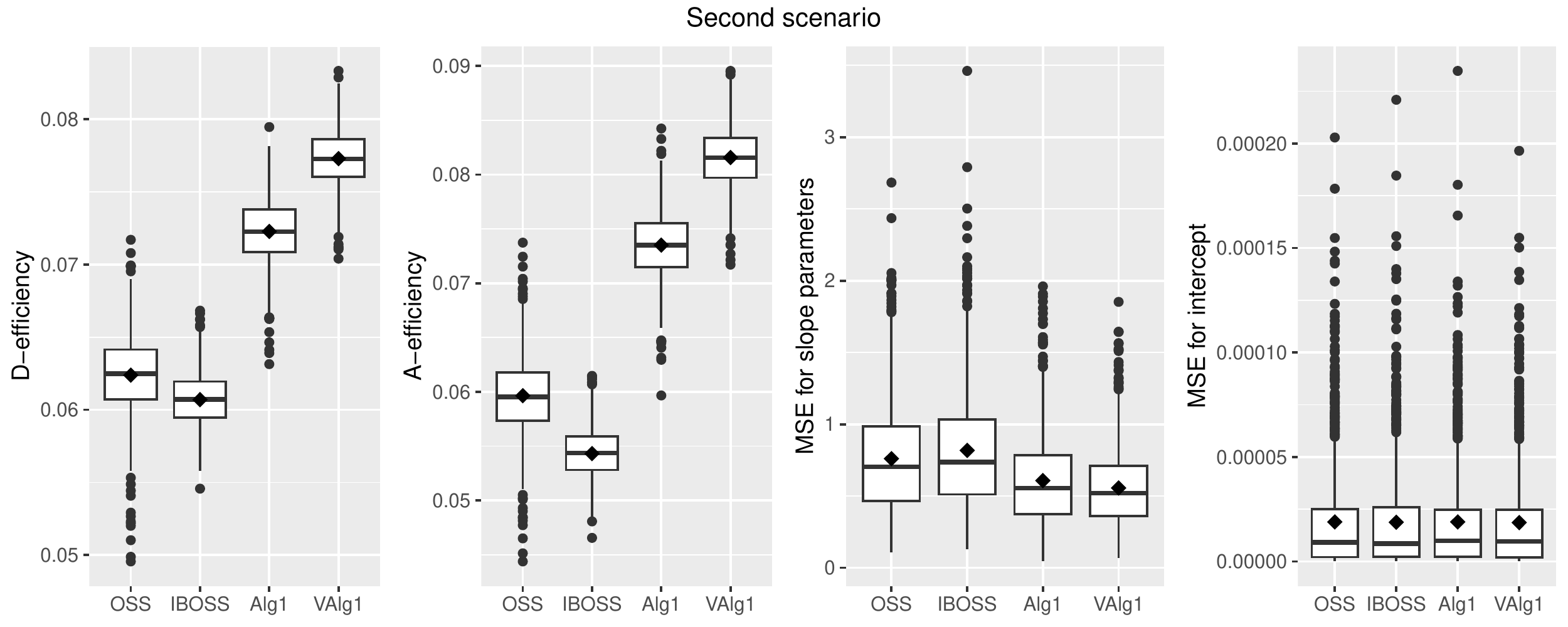}\par
\end{center}
\caption{The MSEs, D- and A-efficiencies for the subdata selected by different approaches, when the full data size is $n=5\times10^5$ in the presence of outliers, the subdata size is $k=100$, and the number of candidate data points is $K=25$. Alg1 is executed with $5$ iterations, and VAlg1 is executed once.}
\label{fig_outliers}
\end{figure}

The results are similar with the ones in Figure \ref{fig1}. In the presence of outliers, both Alg1 and VAlg1 outperform the IBOSS and OSS algorithms in each of the D-, A-optimality criteria, and provide more accurate estimates for the model parameters as one can see in Figure \ref{fig_outliers}.  The important finding is the magnitude of progress in the D-, A-optimality criteria rather the absolute improvement which is sure since we start from the OSS algorithm and we proceed by improving the generalized variance.

\section{Real data applications}
\label{section_data}

We evaluate the performance of our approach on three real data examples. In the first example, the accuracy of the ordinary least squares estimates of slope parameters in model \eqref{model1} is examined, based on the subdata obtained by our approach. In the second example, we examine the trade-offs between execution time and information gain in terms of D-efficiency and A-efficiency, based on the subdata obtained by our approach. In the last example, which was used in \cite{wang2019information}, we
examine the convex hull generated by the subdata of our approach between some selected covariates.

\subsection{Protein tertiary structure}
The dataset of the first example is related to physicochemical properties of protein tertiary structure. The dataset is taken from CASP 5-9. The full data contains $n=45, 730$ data points and we have selected eight measurements: total surface area, non polar exposed area, fractional area of exposed non polar part of residue, molecular mass weighted exposed area, average deviation from standard exposed area of residue, Euclidean distance, secondary structure penalty and spacial distribution constraints, so the number of covariates in the model is $p=8$. The response variable is the size of the residue. Further information about the dataset can be found in ``UCI Machine Learning Repository'' \cite{dua2019}.

We are interested in comparing the performance of our algorithms with the algorithms of the approaches of IBOSS and OSS, and so we consider the MSE for the vector of slope parameters for each algorithm by using $100$ bootstrap samples, as in \cite{wang2019information} and \cite{wang2021oss}. Each bootstrap sample is a random sample of size $n$ from the full data using sampling with replacement. For a bootstrap sample, we implement each algorithm in order to obtain the subdata and then from the selected subdata we estimate the parameters of the model. 

We consider $k=6p, 10p, 16p, 32p$, and $K=20$. Alg1 is executed with $5$ iterations, and VAlg1 is executed once. Both Alg1 and VAlg1 start from the IBOSS approach, since the latter outperforms the OSS approach. Figure \ref{protein} shows the bootstrap MSEs by different approaches. Logarithm with base 10 is taken of each MSE for better presentation of the figures. The mean values ($\blacklozenge$) are also provided. Our approach, that is both Alg1 and VAlg1, outperforms the  IBOSS and OSS algorithms in minimizing the bootstrap MSEs. We notice that Alg1 and VAlg1 have a similar behavior, but VAlg1 being slightly better in this case. Also note that the improvement we get from our algorithms is somehow the same even when the subdata size is getting bigger.

\begin{figure}[!thb]
\begin{center}
\includegraphics[width=1\textwidth]{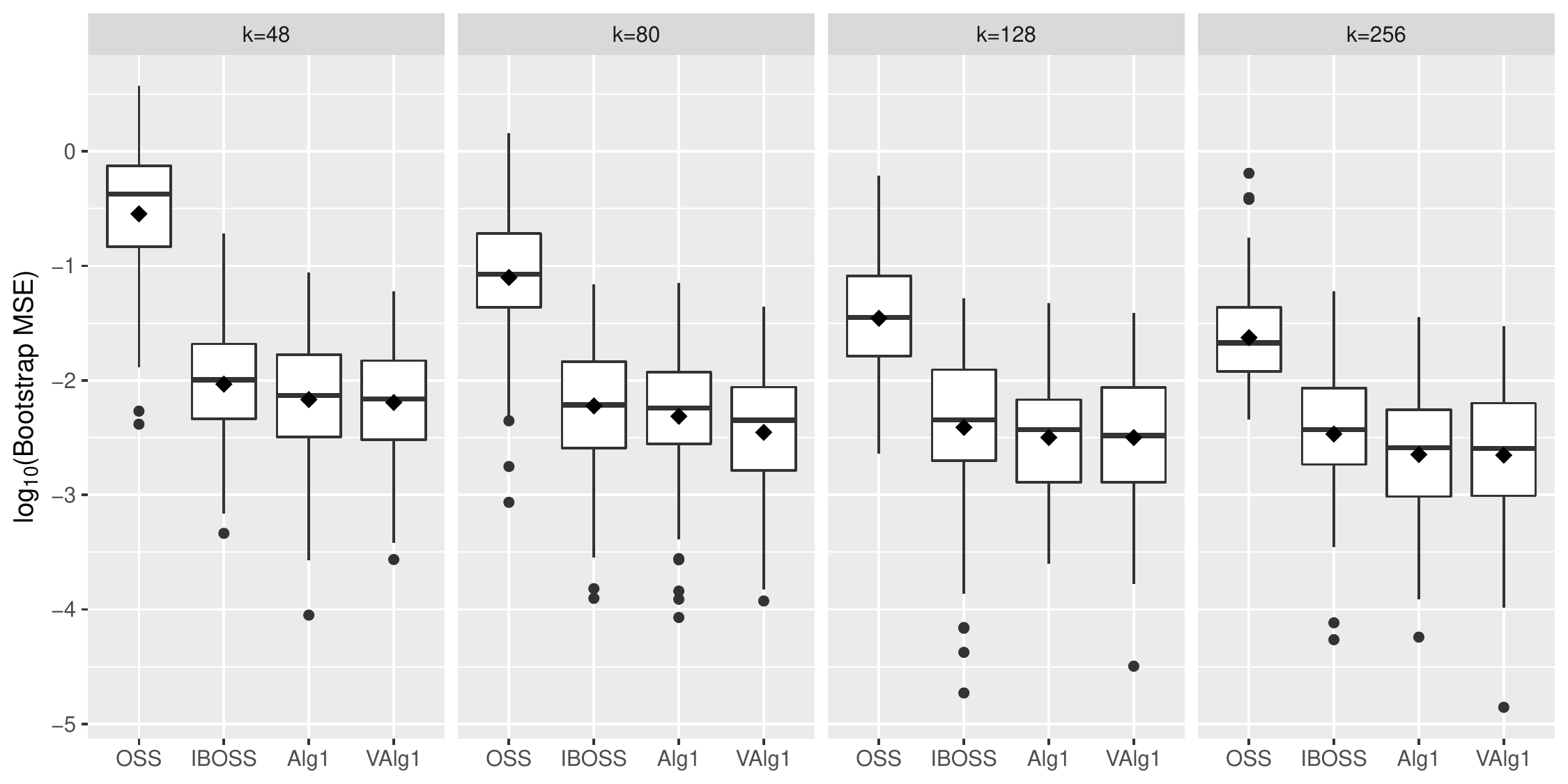}
\end{center}
\caption{The bootstrap MSEs for the subdata selected by different approaches, when the subdata size is $k=6p, 10p, 16p, 32p$, and the number of candidate data points is $K=20$. Alg1 is executed with $5$ iterations, and VAlg1 is executed once.}
\label{protein}
\end{figure}

\subsection{US domestic flights}
The current dataset is related to US domestic flights from 1990 to 2009. The full data contains $n=3, 606, 803$ data points and we have selected five measurements: the number of seats available on flights from origin to destination, the number of flights between origin and destination, the distance flown between origin and destination, origin city's population as reported by US Census, and destination city's population as reported by US Census, so the number of covariates is $p=5$. The response variable is the number of passengers transported from origin to destination. Further information about the dataset can be found in \url{https://www.kaggle.com/datasets/flashgordon/usa-airport-dataset}.

We are interested in comparing the performance of our algorithms with the algorithms of the approaches of IBOSS and OSS, in terms of both D-efficiency and A-efficiency. Also, we focus on the execution times of our algorithms and the approaches of IBOSS and OSS, in order to clarify the trade-offs between information gain and execution time. 

We consider $k=12p, 18p, 28p, 56p$, and $K=20$. Alg1 is executed with $5$ iterations, and VAlg1 is executed once. Both Alg1 and VAlg1 start from the IBOSS approach, since the latter outperforms the OSS approach. Figure \ref{flights} shows the D-efficiency, the A-efficiency and the execution time (in seconds) by different approaches. It is useful to note that with rather negligible additional time we can improve substantially the A-, D-efficiencies with respect to the IBOSS algorithm which is the starting point here. We need to mention that the execution times of both Alg1 and VAlg1 are added to the execution time of the IBOSS algorithm. Also, note that for this dataset the new algorithms have overall time needed close to the one of OSS approach, while the improvement in efficiency is quite large.

\begin{figure}[!thb]
\begin{center}
 \includegraphics[width=1\textwidth]{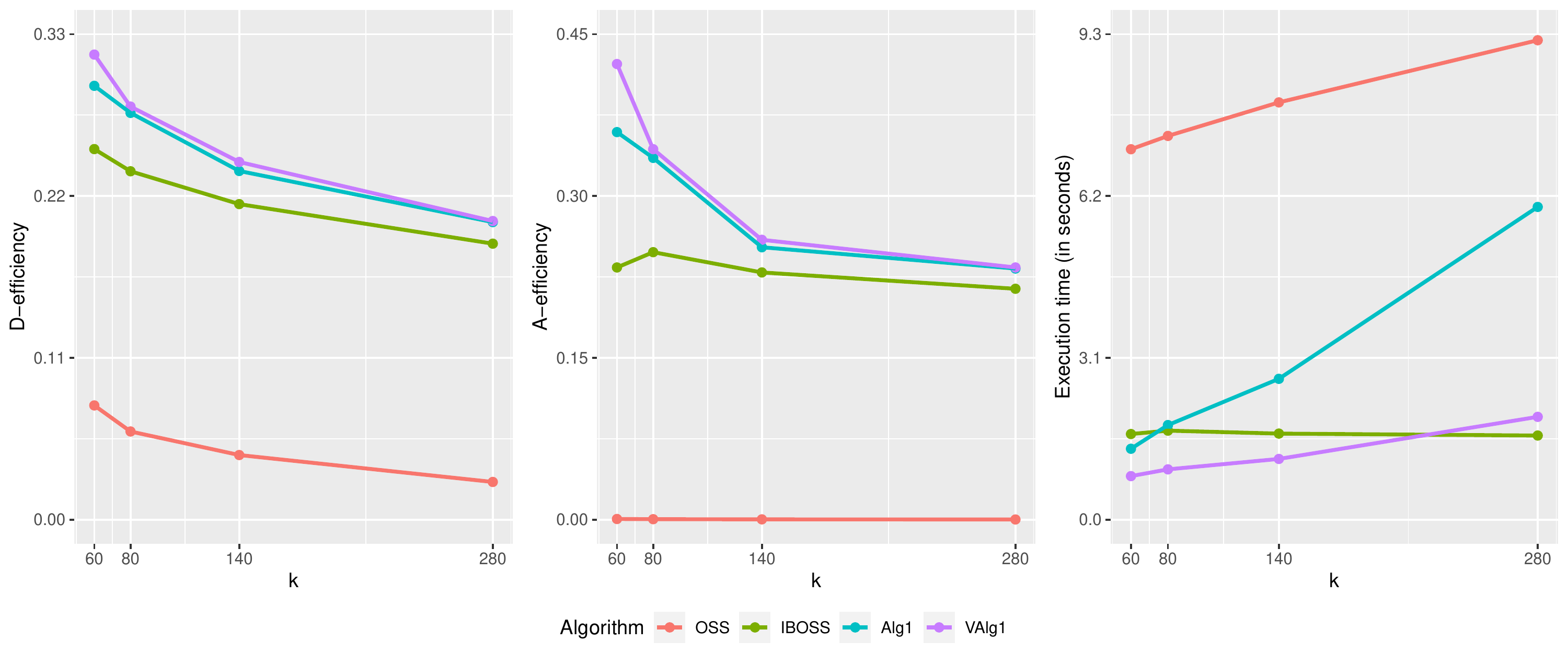}
\end{center}
 \caption{D-efficiency, A-efficiency and execution time (in seconds) for the subdata selected by different approaches, when the subdata size is $k=12p, 18p, 28p, 56p$, the number of candidate data points is $K=20$, Alg1 is executed with $5$ iterations and VAlg1 is executed once.}
 \label{flights}
\end{figure}

\subsection{Chemical sensors data}
In this example, we take into consideration chemical sensors data that were collected in order to develop algorithms for continuously monitoring or improving response time of sensory systems \citep{fonollosa2015}. The data were collected at the ChemoSignals Laboratory
in the BioCircuits Institute, University of California San Diego. The dataset contains the readings of 16 chemical sensors exposed to the mixture of Ethylene and CO at varying concentrations in air. The 16 chemical sensors are of four different types, that is TGS-2600, TGS-2602, TGS-2610, TGS-2620. To be more precise, the order of the 16 sensors is TGS-2602, TGS-2602, TGS-2600, TGS-2600, TGS-2610, TGS-2610, TGS-2620, TGS-2620, TGS-2602, TGS-2602, TGS-2600, TGS2-600, TGS-2610, TGS-2610, TGS-2620 and TGS-2620. Further information about the dataset can be found in \cite{fonollosa2015}.

For illustration, we followed \cite{wang2019information} who used this example, and so we use the readings from the first 15 sensors as covariates, taking a log-transformation of the sensors readings. Also, we do not consider readings from the second sensor, and so there are $p = 14$ covariates in this example. Moreover, the first $20,000$ data points are excluded, and so the full data used contains $n=4,188,261$ data points.

We select $k=140$ data points, in order to compare our approach with the approaches of IBOSS and OSS based on the generated convex hull. As mentioned in Section \ref{introduction}, the convex hull of the selected subdata should be in some sense as close as possible to the one generated by the full data. Also, we consider $K=10$. Alg1 is executed with $5$ iterations, and VAlg1 is executed once. Figure \ref{chull} shows the generated convex hulls between the $9$th and the $3$rd sensor, as well as between the $14$th and the $8$th sensor, by different approaches. In both cases, the convex hull generated by the subdata selected by VAlg1 is the closest one to this generated by the full data. It seems that the second closest convex hull generated by the subdata selected by Alg1. The behavior for the other pairs of covariates is similar. However, when the covariates are strongly correlated, the hull is quite smaller and the improvement that we attain is smaller as well.

\begin{figure}[!thb]
\begin{center}
\includegraphics[width=1\textwidth]{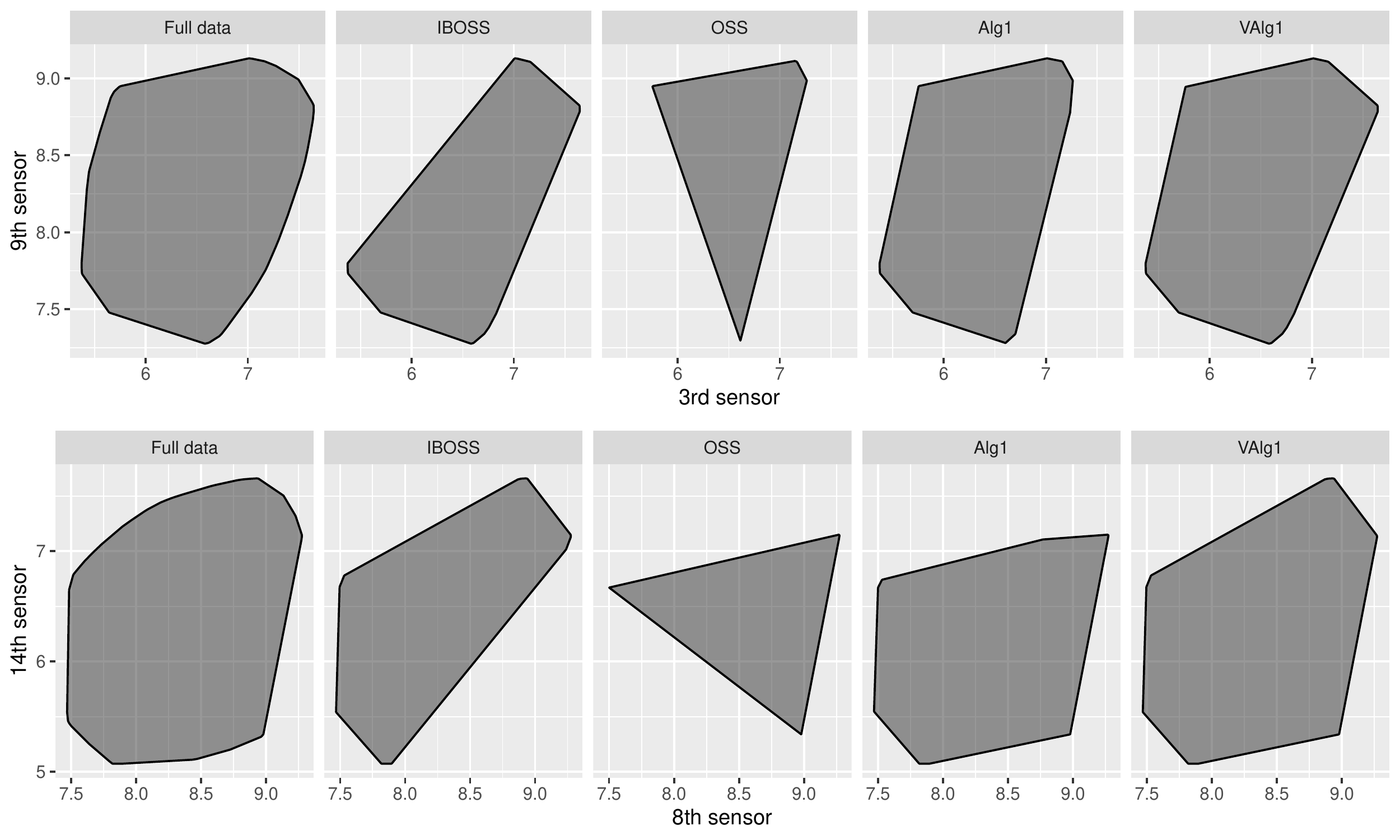}
\end{center}
\caption{The convex hulls between the $9$th and the $3$rd sensor, as well as between the $14$th and the $8$th sensor, for the subdata selected by different approaches, when the subdata seize is $k=140$ and the number of candidate data points is $K=10$. Alg1 is executed with $5$ iterations, and VAlg1 is executed once.}
\label{chull}
\end{figure}

Both Alg1 and VAlg1 started from the OSS approach. Note that the convex hull generated by the IBOSS approach is closer to this generated by the full data, compared with this generated by the OSS approach. However, our approach starts with the OSS approach, and both Alg1 and VAlg1 lead to generate convex hulls closer to this generated by the full data, compared with this generated by the IBOSS approach. This fact gives us an indication that our approach can improve another approach, i.e., the OSS approach, to a great extent, even if the latter seems to select less informative data points compared to another approach, i.e., the IBOSS approach. Similar results can occur if the IBOSS approach selects less informative data points compared to the OSS approach.

Figure \ref{chull} provides an insight into each approach in the case of a real big data example. Our approach dominates other approaches in terms of maximizing the volume generated by the selected subdata, even in the case of an extraordinary large dataset. Also, it is evident that the choice of a subsampling approach is particularly important in cases of real big data, since each approach poses its own risks. Figure \ref{chull} is very revealing about the way the algorithms work. The algorithm of the IBOSS approach selects extreme data points from both covariates, and so it ignores the shape of the data to a great extent. In the plotted case the shape is quite large and the algorithm of the IBOSS approach fails to account for that. The algorithm of the OSS approach selects data points in the corners of a rectangular and because of the data it fails also to see the shape. Our proposed algorithms started from the data points selected by the OSS approach exchange data points in order to increase the convex hull, and hence they end up with a shape closer to the true one of the full data. This way the shaded area is quite large and comes closer to the one of the full data.

\section{Concluding remarks}
\label{concluding_remarks}
We have presented algorithms to select data points in an optimal way from a big dataset so as to be able to run regression and derive coefficients that share as much information as possible. The newly developed algorithms were compared with existing ones to show the kind of improvement that we can take back. 

The theoretical results in Section \ref{theor} indicate the working direction of \cite{wang2019information}, \cite{wang2021oss} and \cite{ren&zhao}. However, as \cite{wang2021oss} stated, finding subdata that exactly approach an OA may be impossible in many cases. Theorem \ref{theorem} shows that
an orthogonal array provides the best result but in practice this is not feasible with real data. The results in Sections \ref{section_simulation} and \ref{section_data} show that neither the selection of only extreme data points nor the approach of an OA can always lead to the most informative subdata, even if the corresponding approaches could be preferred because they are fast. Therefore, our approach is generally based on the maximization of the generalized variance, without being interested in directly obtaining specific data points according to Theorem \ref{theorem}. However, Theorem \ref{theorem} implies that searching towards the maximization is a good strategy for practical purposes.  
 
 Also, \cite{wang2021oss} discussed that the OSS approach could be improved by a greedy modification of their algorithm, that is to interchange data points between the selected subdata and the remaining data points based on the minimization of a discrepancy function. Our approach is an updated version of what \cite{wang2021oss} discussed, since we obtain some candidate data points in advance, and so a very time-consuming algorithm is avoided. Moreover, two features of our approach that can be modified are the number of iterations of Alg1 and which data points are interchanged. At this point, it is important to discuss the value of $K$ which is an even number of data points selected for each covariates. Initially, it is difficult to determine the value of $K$, since the establishment of a general rule for its determination is not feasible. A larger value of $K$ results in a greater number of data points in \textbf{F}. In this scenario, both Alg1 and VAlg1 may select more informative subdata, since the pool of candidate data points for the final subdata is larger. However, it is worth mentioning that a larger value of $K$ can make both Alg1 and Valg1 extremely time-consuming. Consequently, one may consider a trade-off between the improvement of the final subdata and the execution times of both Alg1 and VAlg1. An optimal value of $K$ is difficult to be determined in advance. Therefore, we recommend a moderate value of $K$, typically ranging between $10$ and $25$, ensuring that the maximum number of data points in \textbf{F} lies between $10p$ and $25p$.
 
 In fact, the algorithms of our approach are greedy extension to existing ones and hence one must consider a trade-off between improving the selected data points at the cost of some additional effort time needed. In practice, we suggest that starting with the OSS approach can be a secure choice, as it has generally demonstrated better performance than the IBOSS approach. While here we pursue D-optimality, the presented approach also satisfies good properties with respect to A-optimality. Additionally, extensions to other models are straightforward, as for example GLM type of models.

\bibliographystyle{plainnat}
\bibliography{main}

\end{document}